\documentclass[12pt]{elsarticle}
\usepackage[margin=2.5cm]{geometry}% by courtesy of Mico
\usepackage{lipsum}

\makeatletter
\def\ps@pprintTitle{%
   \let\@oddhead\@empty
   \let\@evenhead\@empty
   \def\@oddfoot{\reset@font\hfil\thepage\hfil}
   \let\@evenfoot\@oddfoot
}
\makeatother
\usepackage{amsfonts,color,morefloats,pslatex}
\usepackage{amssymb,amsthm, amsmath,latexsym}

\newtheorem{theorem}{Theorem}
\newtheorem{lemma}[theorem]{Lemma}
\newtheorem{proposition}[theorem]{Proposition}
\newtheorem{corollary}[theorem]{Corollary}

\newtheorem{open}[theorem]{Open Problem}
\newtheorem{definition}[theorem]{Definition}
\newtheorem{example}[theorem]{Example}

\newtheorem{remark}[theorem]{Remark}

\newcommand{\gf}{{\mathrm{GF}}}
\newcommand{\C}{{\mathcal{C}}}

\newcommand{\supp}{{\mathrm{Supp}}}
\newcommand{\bc}{{\mathbf{c}}}

\usepackage{blindtext}

\begin{document}
%\tableofcontents

\begin{frontmatter}



\title{Full Characterization of   Minimal Linear Codes as Cutting Blocking Sets
\tnotetext[fn1]{C. Ding's research was supported by the Hong Kong Research Grants Council,
Proj. No. 16300418. C. Tang was supported by National Natural Science Foundation of China (Grant No.
11871058) and China West Normal University (14E013, CXTD2014-4 and the Meritocracy Research
Funds).}
}

\author[cmt]{Chunming Tang}
\ead{tangchunmingmath@163.com}

\author[cmt]{Yan Qiu}
\ead{yanqiucwnu@163.com}

\author[qyl]{Qunying Liao}
\ead{qunyingliao@sicnu.edu.cn}

\author[zcz]{Zhengchun Zhou}
\ead{zzc@swjtu.edu.cn}
%\cortext[lcj]{Corresponding author}
\address[cmt]{School of Mathematics and Information, China West Normal University, Nanchong, Sichuan,  637002, China}
\address[qyl]{Institute of Mathematics and Software Science, Sichuan Normal University, Chengdu, 610031, China}
\address[zcz]{School of Mathematics, Southwest Jiaotong University, Chengdu, 610031, China}
%\address[vdt]{
%Department of Mathematical Sciences, Michigan Technological University,
%Houghton, Michigan 49931, USA
%}

%\author{Cunsheng Ding}
%\ead{cding@ust.hk}

%\cortext[lcj]{Corresponding author}
%\address{Department of Computer Science and Engineering,
%The Hong Kong University of Science and Technology,
%Clear Water Bay, Kowloon, Hong Kong, China}

%\tableofcontents

\begin{abstract}
In this paper, we first study in detail the relationship between minimal linear codes and cutting blocking sets,
which were recently introduced by
 Bonini and Borello, and  then completely characterize minimal linear  codes  as  cutting blocking sets.
 As a direct result, minimal projective  codes of dimension $3$ and
 $t$-fold blocking sets  with $t\ge 2$ in  projective planes  are identical objects.
Some bounds on the parameters of minimal codes are derived from this characterization.
This confirms a recent conjecture by Alfarano, Borello and Neri in
[a geometric characterization of minimal codes and their asymptotic performance, arXiv:1911.11738, 2019] about
 a lower bound of the minimum distance of a minimal code.
Using this new link between minimal codes and blocking sets, we
also present new general  primary and secondary constructions of minimal linear codes.
 As a result, infinite families of minimal linear codes not satisfying the  Aschikhmin-Barg's condition are obtained.
 In addition to this, the weight distributions of two subfamilies of the proposed minimal linear codes are established.
 Open problems are also presented.
\end{abstract}


\end{frontmatter}

%\tableofcontents

\section{Introduction}
Throughout the paper, we will assume the reader to have familiarity with linear codes ( see for instance \cite{MS77}). A $q$-ary linear code of
 length $n$ and dimension $k$ will be referred to as an $[n, k]_q$ code. Further, if the code has minimum distance $d$, it will be referred to as an $[n, k, d]_q$ code. When the alphabet size $q$ is clear from the context, we omit the subscript.
Let $\C$ be an $[n, k, d]_q$ linear code. $\C$ is called projective if any two of its coordinates are linearly independent,
or in other words, if the minimum distance $d^{\perp}$ of its dual code $\C^{\perp}$ is at least three.

The Hamming weight (for short, weight) of a vector $\mathbf v$ is the number of its nonzero entries and is denoted $\mathrm{wt}(\mathbf v)$.
The minimum (respectively, maximum) weight of the code $\C$ is the minimum (respectively, maximum) nonzero weight among all codewords of $\C$,
$w_{\min} = \min (\mathrm{wt}(\mathbf c))$ (respectively, $w_{\max} = \max (\mathrm{wt}(\mathbf c))$).

Let $\bc=(c_0, \cdots, c_{n-1})$  be a codeword in $\C$.
The \emph{support} $\supp(\bc)$ of the codeword $\bc$  is the set
of indices of its nonzero coordinates:
$$\supp(\bc)=\{i: c_i \neq 0\}.$$
A codeword $\bc$ of the linear code $\C$ is called minimal if its support
does not contain the support of any other linearly independent codeword.
$\C$ is called a minimal linear code  if all codewords of $\C$  are minimal.
Minimal codes are a special class of linear codes.
They have applications in secret sharing schemes \cite{Mas93,Mas95}.
A sufficient condition for a linear code to be minimal is given in the following lemma \cite{AB98}.
\begin{lemma}[Aschikhmin-Barg]\label{lem:AB}
 A linear code $\C$ over $\gf(q)$ is minimal if $\frac{w_{\min}}{w_{\max}} >\frac{q-1}{q}$.
\end{lemma}

Many minimal linear codes satisfying  the  condition $\frac{w_{\min}}{w_{\max}} >\frac{q-1}{q}$
are obtained from linear codes with few weights \cite{DD15,DY03,M17,MOS17,MOS18,TLQZT}.

The sufficient condition in Lemma \ref{lem:AB} is not necessary for minimal codes. Recently,
searching for minimal linear codes with $\frac{w_{\min}}{w_{\max}} \le  \frac{q-1}{q}$
 has been an interesting research topic.
Chang  and Hyun \cite{CY18} made a breakthrough and constructed an infinite family of minimal binary
linear codes with $\frac{w_{\min}}{w_{\max}} < \frac{1}{2}$. Ding, Heng and Zhou \cite{DHZI} gave
a necessary and sufficient condition for a binary linear code to be minimal.
Three infinite families of minimal binary linear codes with $\frac{w_{\min}}{w_{\max}} \le  \frac{1}{2}$ were obtained using this condition.
They also constructed an infinite family of minimal ternary linear codes with with $\frac{w_{\min}}{w_{\max}} \le  \frac{2}{3}$ in \cite{DHZF}.
Bartoli and Bonini \cite{BBI} generalized the construction of minimal linear codes in \cite{DHZF} from ternary  case to
odd characteristic case. In \cite{BBG}, an inductive construction of minimal codes was presented.
Li and Yue \cite{LY} obtained some minimal binary linear codes with Boolean functions. Xu and Qu \cite{XQ19} constructed
minimal $q$-ary linear codes from some special functions. Lu, Wu and Cao \cite{LWC} studied the existence of  minimal linear codes.
Bonini and Borello \cite{BB19} presented a family of codes arising from cutting blocking sets.
Infinitely many of these codes do not satisfy Ashikhmin-Barg's condition.

In this paper, we mainly study further the characterizations and  constructions  of minimal codes.
First, we investigate in detail the relationship between minimal linear codes and  blocking sets,
 and completely characterize minimal linear codes as cutting blocking sets. In particular,
 minimal projective codes of dimension $3$ and $t$-fold blocking sets with $t \ge  2$ in projective planes are identical objects.
 A tight lower bound for the minimum distance of a minimal code was derived using the geometric characterization of minimal codes.
This settles the conjecture in \cite{ABN}.
 By the new characterization of minimal linear codes, we present
 a primary construction and a general secondary construction of minimal codes.
 Some new infinite classes of minimal $q$-ary linear codes with $\frac{w_{\min}}{w_{\max}} \le \frac{q-1}{q}$ are derived.
Finally we determine the weight distributions of two subfamilies of the proposed minimal codes.
Open problems are also presented.

The rest of this paper is organized as follows. In Section \ref{sec:intr}, we present some basic results
on linear codes from defining sets and blocking sets. In Section  \ref{sec:char},
we  study in detail the relationship between minimal linear codes and blocking sets.
It enables us to identify a minimal code as a cutting blocking set.
In Section \ref{sec:constr}, we present a primary construction of minimal codes from hyperplanes and a
general secondary construction of minimal codes, and establish the weight distributions of two subfamilies of the proposed minimal codes.
In section \ref{sec:conc}, we conclude this paper.

\section{Background}\label{sec:intr}
\subsection{Linear codes from multisets in vectorial  spaces}
Let $\mathbb V$ be a vector space over $\gf(q)$ and let $\langle \cdot, \cdot  \rangle$
be an inner product in $\mathbb V$.
For $\mathbf{v} \in \mathbb V \in \setminus \{\mathbf{0}\}$ we will denote by $\langle \mathbf  v \rangle$ the
one dimensional subspace generated by $\mathbf v$.
A multiset $D$ in $\mathbb V$  is \emph{$k$-dimensional} if the
linear subspace $\mathrm{Span}(D)$ over $\gf(q)$ generated by $D$ has dimension $k$.
A subset $\overline{D}$  of $\mathbb V\setminus \{\mathbf 0\}$ is called a projection of $D$
if for any nonzero $\mathbf v \in D$ there exists a unique $\mathbf v'\in \overline{D}$ such that
$\mathbf v =\lambda \mathbf v'$ where $\lambda \in \gf(q)^*$.
And $D$ is called projective if $\# D= \# \overline{D}$.

Let $D:=\{\{ \mathbf g_1, \cdots, \mathbf g_n\}\}$ be a multiset in $k$-dimensional
vector space $\mathbb V$. A classic generic construction
of linear codes from multisets of vector spaces is described  as
\begin{align}\label{def:C-D}
\C_{D}=\{\left (\langle \mathbf v, \mathbf g_0 \rangle, \cdots, \langle \mathbf v, \mathbf g_{n-1} \rangle\right ): \mathbf v \in \mathbb V\}.
\end{align}
Empoying this general scheme, Ding et al. constructed many families of linear codes with few weights \cite{DLN08,DN07}.
We call $D$  the defining set of $\C_{D}$.  By definition, the dimension of the code $\C_{D}$ is at most $k$.
Although different orderings of the elements of $D$ give different linear codes, these codes are permutation equivalent.
Hence we do not consider these codes obtained by different orderings of the elements in $D$.
Further, $\C_{D}$ is a projective code if and only if $D$ is projective.
Xiang proved that any $q$-ary linear code $\C$ may be generated with a defining set $D$ via this construction \cite{Xiang16}.

The familiar Griesmer bound says that for an $[n,k,d]$ code $\C$ over
$\gf(q)$,
\begin{eqnarray}\label{eq:Griesmer}
n\ge d+ \left \lceil \frac{d}{q}  \right \rceil + \cdots \left \lceil \frac{d}{q^{k-1}}  \right \rceil.
\end{eqnarray}
The bound was proved by Griesmer in 1960 for binary linear codes and generalized by Solomon and Stiffler in 1965.

\subsection{Blocking sets and blocking multisets}
In \cite{BB19}, Bonini and Borello introduced the concept of cutting  $s$-blocking sets for sets in affine, projective and vector spaces.
In order to  study minimal codes, we will consider the corresponding concept for multisets.
For any  multiset $D$ in a vector space $\mathbb V$, we will denote the multiset $D \setminus \{\mathbf 0\}$
by $D^*$.

In geometry,  a blocking set is a set of points in a projective plane which every line intersects and
which does not contain an entire line. Instead of talking about projective planes and lines,
one could deal with  high dimensional spaces and their subspaces.
\begin{definition}
Let $D$ be a multiset of an $n$-dimensional vector space $\mathbb V$.
Then $D$ is called a \emph{vectorial  $s$-blocking multiset} if
every  subspace of codimension $s$  of  $\mathbb V$ has a non-empty intersection with $D^*$.
Furthermore, if $D$ is a set, $D$ is also called a \emph{vectorial  $s$-blocking set}.
A \emph{vectorial $1$-blocking multiset} is also referred  as  vectorial blocking multiset.
\end{definition}

A generalization of blocking sets, called multiple blocking sets, was introduced by Bruen in \cite{Bruen86}.
\begin{definition}
A vectorial  $s$-blocking multiset $D$ of a vector space $\mathbb V$ is called \emph{$t$-fold}
if every $(n-s)$-dimensional subspace contains at least $t$ elements of $D^*$ and some $(n-s)$-dimensional subspace contains exactly $t$ elements of $D^*$.
\end{definition}

In order to investigate minimal codes, Bonini and Borello introduced the crucial concept of cutting blocking sets \cite{BB19}.
\begin{definition}\label{def:cutting}
A vectorial $s$-blocking multiset $D$ in a vector space $\mathbb V$ is \emph{cutting} if
its intersection with every linear subspace of codimension $s$ of $\mathbb V$ is not contained in any
other linear subspace of codimension $s$.
\end{definition}

With some effort, we can give similarly the definitions of affine cutting  $s$-blocking sets and
projective cutting  $s$-blocking sets; and the details for these can be found in \cite{BB19}.

According to the previous definitions, $D$ is a vectorial cutting $s$-blocking multiset if and only
if $\overline{D}$ is a vectorial cutting $s$-blocking set, where $\overline{D}$ is
a projection of $D$. In fact, if $\overline{D}$ is a vectorial cutting $s$-blocking set in an
$n$-dimensional vector space $\mathbb V$ over $\gf(q)$, then $\overline{D}$
can be identified as a projective cutting $s$-blocking set in the $(n-1)$-dimension projective
space $\mathrm{PG}(n-1, q)$.  If $D$ is  a projective cutting $s$-blocking set in $\mathrm{PG}(n-1, q)$,
then so is $D'$, where $D\subseteq D'$. Some authors refer to
a $t$-fold blocking set of cardinality $f$
in $\mathrm{PG}(n-1, q)$  as an $\{f,t; n-1,q\}$-\emph{minihyper}.

Bose and Burton determined the smallest point sets of $\mathrm{PG}(n, q)$ that meet every subspace of $\mathrm{PG}(n, q)$ of a given dimension $(n-s)$.
\begin{theorem}[Bose and Burton, \cite{BB67}]
An $s$-blocking set of $\mathrm{PG}(n, q)$ has at least $\theta_{s}:=\frac{q^{s+1}-1}{q-1}$ points.
In the case of equality the $s$-blocking set is an $s$-dimensional subspace.
\end{theorem}

An $s$-blocking set containing an $s$-dimensional subspace is called trivial. For a non-trivial blocking set
in the projective plane $\mathrm{PG}(2,q)$ the theorem gives  an improved lower bound.
\begin{theorem}[Bruen, \cite{Bruen70,Bruen71}]
In $\mathrm{PG}(2, q)$ a non-trivial blocking set has size at least $q+\sqrt{q}+1$. In the case of equality the blocking set is a Baer subplane.
\end{theorem}
A Baer subplane of $\mathrm{PG}(2, q)$ is an embedded $\mathrm{PG}(2, \sqrt{q})$ subgeometry.

The following theorem follows form a simple counting argument.
\begin{theorem}[Harrach, Proposition 1.5.3, \cite{Harrach13}]
If $D$ is any blocking set of the projective plane $\mathrm{PG}(2,q)$, then for any line $\ell$ not contained in $D$, we have $\# (D \setminus \ell ) \ge q$.
\end{theorem}

The following theorems give some general lower bounds on the size of a $t$-fold blocking set in $\mathrm{PG}(2,q)$.
\begin{theorem}[Ball, \cite{Ball96}]
Let $D$ be a $t$-fold blocking set in $\mathrm{PG}(2, q)$. If $D$ contains no line, then it has at least $tq+\sqrt{tq}+1$ points.
\end{theorem}

\begin{theorem}[Bruen, \cite{Bruen92}]
Let $D$ be a $t$-fold blocking set in $\mathrm{PG}(2, q)$ that contains a line. If $t\ge 2$, then $\#D \ge  tq+q - t+2$.
\end{theorem}

\begin{theorem}[Ball, \cite{Ball96}]
Let $D$ be a nontrivial $2$-fold blocking set in $\mathrm{PG}(2,q)$.  Then the following are true.
\begin{enumerate}
\item If $q<9$ then $D$ has at least $3q$ points.
\item If $q=11,13,17$ or $19$ then $\#D\ge \frac{5q+7}{2}$.
\item If $q=p^{2e+1}>19$ then $\#D \ge 2q+p^e\left \lceil \frac{p^{e+1}+1}{p^e+1} \right \rceil+2$.
\item If $q>4$ is a square then $\# D \ge 2q+2\sqrt{q}+2$.
\end{enumerate}

\end{theorem}

The following result gives a depiction of  $2$-fold blocking sets in projective plane.

\begin{theorem}[Blokhuis, Storme and Sz\"onyi, \cite{BSS99}]
Let $D$ be a $t$-fold blocking set in $\mathrm{PG}(2,q)$, $q=p^e$, $p$ prime, of size $t(q+1)+c$. Let $c_2=c_3=2^{-1/3}$ and $c_p=1$ for $p > 3$.
If $q>4$ is a square, $t<q^{1/4}/2$ and $c<c_p q^{2/3}$, then $c\ge t\sqrt{q}$ and $D$ contains the union of $t$ disjoint Baer subplanes,
except for $t=1$, in which case $D$ contains a line or a Baer subplane.
\end{theorem}

The following theorem presents a lower bound on the size of an affine blocking set in affine space.
\begin{theorem}[Jamison, \cite{Jamison77}]\label{thm:Jamison77}
If $\mathbb V$ is a vector space of dimension $k$ over the finite field $\gf(q)$, then any subset of $\mathbb V$ which meets
every hyperplane of $\mathbb V$ contains at least $k (q - 1)+ 1$ points.
\end{theorem}

\section{Characterization of minimal codes using cutting blocking sets}\label{sec:char}

First of all, we present an equivalent description of cutting blocking sets.
\begin{proposition}\label{prop:cutting-fold}
Let $D$ be a projective subset of a $n$-dimensional vector space $\mathbb V$.
Then $D$ is a cutting blocking set if and only if for any $(n-1)$-dimensional linear
subspace $H$, the intersection
$D \cap H$  is $(n-1)$-dimensional.
\end{proposition}
\begin{proof}
Let $D$ be a cutting blocking set. Assume that $D \cap H$ is $k$-dimensional,
 where $H$ is $(n-1)$-dimensional linear subspace and $k<(n-1)$.
Then there exists a basis $\{ \mathbf v_1, \cdots, \mathbf v_n \}$ of $\mathbb V$
such that $\mathrm{Span}(D \cap H)=\mathrm{Span}(\mathbf v_1, \cdots, \mathbf v_k)$
and $H=\mathrm{Span}(\mathbf v_1, \cdots, \mathbf v_{n-2}, \mathbf v_{n-1})$.
Let $H'$ be the $(n-1)$-dimensional vectorial subspace spanned by $\mathbf v_1, \cdots, \mathbf v_{n-2}$ and $  \mathbf v_{n}$.
Since $k\le (n-2)$, $H$ and $H'$ are different $(n-1)$-dimensional subspaces.
On the other hand, it's easily observed that
\begin{align*}
D \cap H= D \cap \mathrm{Span}(\mathbf v_1, \cdots, \mathbf v_k) \subseteq D \cap H'.
\end{align*}
This is contrary to the definition of cutting blocking sets.

Conversely, assume that for any $(n-1)$-dimensional linear
subspace $H$, the intersection
$D \cap H$  is $(n-1)$-dimensional.
Then, $D$ is a blocking set.
We only need to prove that $D$ is cutting.
Suppose $D$ is not cutting. Then, there are two different
$(n-1)$-dimensional subspaces $H$ and $H'$ such
that $D\cap H \subseteq D \cap H'$.
Thus, $D\cap H$ is contained in the $(n-2)$-dimensional vectorial subspace  $\subseteq H \cap H'$.
This is contrary to the assumption that $D\cap H$ is $(n-1)$-dimensional.
This completes the proof.
\end{proof}

We can now state the main theorem, i.e., the characterization of minimal codes in terms of cutting blocking sets,
which shows that minimal codes and vectorial cutting blocking multisets are identical objects.
Similar results have been obtained independently and simultaneously by Alfarano, Borello and Neri in \cite{ABN}.
\begin{theorem}\label{thm:characterization}
 Let $\C$ be a $q$-ary $[n,k]$ linear code with generator matrix $G=[\mathbf g_0, \cdots, \mathbf g_{n-1}]$, where $\mathbf g_i \in \gf(q)^k$.
 Let $\overline{D}$ denote any projection of  the multiset $D=\{\{ \mathbf g_0, \cdots, \mathbf g_{n-1} \}\}$. Then, $\C$ is a minimal code if and only if
  $\overline{D}$ is a vectorial  cutting blocking set in the $k$-dimensional vector space $\gf(q)^k$, in other words,
   $\overline{D}$ is a projective  cutting blocking set in $\mathrm{PG}(k-1, q)$.
\end{theorem}
\begin{proof}
Since $\mathcal C$ is a minimal code if and only if $\mathcal C_{\overline{D}}$
is a minimal code, we only need to consider the case that $\C$ is projective.
When $\mathcal C$ is a projective code, we can choose $\overline{D}=D$.
 Notice that for any codeword $\bf c \in \C$,
there exists a unique $\mathbf{v} \in \gf(q)^k$ such that
\begin{align}\label{eq:c_v}
\mathbf c=\mathbf c_{\mathbf v}:=\left (\langle \mathbf v, \mathbf g_0 \rangle, \cdots, \langle \mathbf v, \mathbf g_{n-1} \rangle \right ).
\end{align}
Then
\begin{align}\label{eq:Supp-H}
\mathrm{Supp}(\mathbf c)= \left \{0, 1, \cdots, n-1 \right \} \setminus \left \{i: \mathbf g_i \in H_{\mathbf v} \cap D \right \},
\end{align}
where $H_{\mathbf v}$ is the $(k-1)$-dimensional vectorial subspace $\{\mathbf w \in \gf(q)^k: \langle \mathbf v , \mathbf w \rangle =0\}$.

Let $\C$ be a minimal code. Assume that $D$ is not a cutting blocking set. Then there exist two nonzero vectors
$\mathbf v, \mathbf v'$ such that $ H_{\mathbf v} \cap D \subseteq H_{\mathbf v'} \cap D$ and $\mathbf v \not \in \left \langle  \mathbf v'  \right \rangle$.
Hence $\mathrm{Supp}(\mathbf c') \subseteq \mathrm{Supp}(\mathbf c)$ and $\mathbf c \neq \lambda \mathbf c'$ for any $\lambda \in \gf(q)$
, where $\mathbf c= \left (\langle \mathbf v, \mathbf g_0 \rangle, \cdots, \langle \mathbf v, \mathbf g_{n-1} \rangle \right )$
and $\mathbf c'= \left (\langle \mathbf v', \mathbf g_0 \rangle, \cdots, \langle \mathbf v', \mathbf g_{n-1} \rangle \right )$.
This is contrary to the assumption that $\C$ is a minimal code.

Conversely, let  $\overline{D}$ be a cutting blocking set in the $k$-dimensional vector space $\gf(q)^k$.
Suppose that $\C$ is not a minimal code. Then there exist two nonzero
$\mathbf v, \mathbf v' \in \gf(q)^k$ such that $\mathrm{Supp}(\mathbf c_{\mathbf v}) \subseteq \mathrm{Supp}(\mathbf c_{\mathbf v'})$
and $\mathbf v  \neq \lambda \mathbf v'$ for any $\lambda \in \gf(q)$
, where $\mathbf c_{\mathbf v}$ and $\mathbf c_{\mathbf v'}$ are given by (\ref{eq:c_v}).
Applying (\ref{eq:Supp-H}), one
obtains $H_{\mathbf v'} \cap D \subseteq H_{\mathbf v} \cap D$. This is contrary to the assumption that $D$ is cutting.
This completes the proof.

\end{proof}

\begin{example}
Let $D_{\le h}$ be the subset of $\gf(2)^k$ defined by
\begin{align*}
D_{\le h}=\left \{ \mathbf x \in  \gf(2)^k : 1\le \mathrm{wt}(\mathbf x)\le h  \right  \},
\end{align*}
where $k$ and $h$ are integers such that $k \ge  4$ and $2 \le h \le k$.
Then, it was shown that $\C_{D_{\le h}}$ is a minimal linear code \cite{ZYW19}.
From Theorem \ref{thm:characterization}, $D_{\le 2}$ is a vectorial cutting blocking set in $\gf(2)^k$.
Hence, for any $h\ge 2$,  $D_{\le h}$ is also a vectorial cutting blocking set in $\gf(2)^k$ as $D_{\le 2} \subseteq D_{\le h}$.
By Theorem \ref{thm:characterization} again, $\C_{D_{\le h}}$ is a minimal linear code.
\end{example}

\begin{example}
For $q$ odd, let $D_{\ge h}$ be the subset of $\gf(q)^k$ defined by
\begin{align*}
D_{\ge h}=\left \{ \mathbf x \in  \gf(q)^k : \mathrm{wt}(\mathbf x)\ge h  \right  \},
\end{align*}
where $k$ and $h$ are integers such that $k \ge  4$ and $1 \le h \le k-1$.
Then, it was shown that $D_{\ge h}$ is a vectorial cutting blocking set in $\gf(q)^k$ \cite{BB19}. It follows from Theorem \ref{thm:characterization} that $\C_{D_{\ge h}}$ is a minimal linear code
and $\C_{\overline{D}_{\ge h}}$ is a minimal projective code.
\end{example}

\begin{example}
Let $k =h \ell$ be a positive integer and consider the subset of $\gf(q)^k$ defined by
\begin{align}
D_{h, \ell}=\left \{ (x_{1}, \cdots, x_{n})\in  \gf(q)^k\setminus \{\mathbf{0}\}: \sum_{j=0}^{\ell -1} x_{hj+1}x_{hj+2} \cdots x_{hj+h} =0 \right \},
\end{align}
where $h$ and $\ell$ are integers such that $h \ge  2$ and $\ell  \ge 2$.
It follows from \cite[Theorem 5.1]{BB19} and Theorem \ref{thm:characterization} that $\C_{D_{h, \ell}}$ is a minimal linear code
and $\C_{\overline{D}_{h, \ell}}$ is a minimal projective code.

\end{example}

Combining Proposition \ref{prop:cutting-fold} and Theorem \ref{thm:characterization} yields the following results.
\begin{corollary}\label{cor:mc-tfoldbs}
 Let $\C$ be a $q$-ary $[n,k]$ minimal linear code with generator matrix $G=[\mathbf g_0, \cdots, \mathbf g_{n-1}]$, where $\mathbf g_i \in \gf(q)^k$.
 Let $\overline{D}$ denote any projection of  the multiset $D=\{\{ \mathbf g_0, \cdots, \mathbf g_{n-1} \}\}$. Then
  $\overline{D}$ is a $t$-fold blocking set in $\mathrm{PG}(k-1, q)$ with $t\ge (k-1)$.
\end{corollary}

\begin{theorem}\label{thm:dim3-MC}
 Let $\C$ be a $q$-ary $[n,3]$ linear code with generator matrix $G=[\mathbf g_0, \cdots, \mathbf g_{n-1}]$, where $\mathbf g_i \in \gf(q)^3$.
 Let $\overline{D}$ denote any projection of  the multiset $D=\{\{ \mathbf g_0, \cdots, \mathbf g_{n-1} \}\}$. Then, $\C$ is a minimal code, if and only if,
  $\overline{D}$ is a $t$-fold blocking set in $\mathrm{PG}(2, q)$ with $t\ge 2$.
\end{theorem}

Theorem \ref{thm:dim3-MC} shows that
minimal projective  codes of dimension $3$ and   $t$-fold blocking sets  with $t\ge 2$ in  projective planes  are identical objects.
It is worth noting that, in the projective space $\mathrm{PG}(k-1,n)$ with $k\ge 4$, the concept of cutting blocking sets
 is stronger than the concept of $(k-1)$-fold blocking sets. Therefore, the result of Theorem \ref{thm:dim3-MC} can not be extended to
 the case $k\ge 4$.

\begin{corollary}
 Let $\C$ be a $q$-ary $[n,k]$ minimal projective linear code  with  maximum weight $w_{\max}$ and
 generator matrix $G=[\mathbf g_0, \cdots, \mathbf g_{n-1}]$, where $\mathbf g_i \in \gf(q)^k$. Then
  ${D}=\left \{\mathbf g_0, \cdots, \mathbf g_{n-1} \right \}$ is a $t$-fold blocking set in $\mathrm{PG}(k-1, q)$ with $t=n-w_{max}$.
\end{corollary}

\begin{corollary}
Let $\C$ be a $q$-ary $[n,k]$ minimal  linear code  of  maximum weight $w_{\max}$. Then
$w_{max}\le n-k+1$.
\end{corollary}

Using the link between minimal codes and blocking sets we then deduce our lower bound
on the minimum distance of a $q$-ary linear code with prescribed dimension.
This confirms a recent conjecture by Alfarano, Borello and Neri in \cite{ABN}.
\begin{theorem}\label{thm:abs-bound}
Let $\C$ be a minimal linear codes of dimension $k$ over $\gf(q)$ and let $d$ be the minimum distance
of $\C$. Then
$$d \ge (q-1)(k-1)+1.$$
\end{theorem}

\begin{proof}
It suffices to prove the theorem in the case when $\C$ is a projective code.
Without loss of generality we can assume that the vector $(0, 0, \cdots, 0, 1, 1, \cdots, 1)$ is a minimum-weight codeword in $\C$.
Thus there exists a generator matrix of $\C$ under the form:
\begin{eqnarray*}
G=\left[
\begin{array}{cccc|cccc}
0 & 0 & \cdots & 0 & 1 & 1 & \cdots &1 \\
\hline
u_{1,1} & u_{1,2} & \cdots & u_{1,n-d} & v_{1,1} & v_{1,2} & \cdots & v_{1,d}\\
u_{2,1} & u_{2,2} & \cdots & u_{2,n-d} & v_{2,1} & v_{2,2} & \cdots & v_{2,d}\\
\vdots  & \vdots & \cdots & \vdots & \vdots & \vdots & \cdots & \vdots \\
u_{k-1,1} & u_{k-1,2} & \cdots & u_{k-1,n-d} & v_{k-1,1} & v_{k-1,2} & \cdots & v_{k-1,d}\\
\end{array}
\right],
\end{eqnarray*}
where $n$ is the length of the linear code $\C$.
Write $\mathbf{u}_i= (u_{1,i},  \cdots, u_{k-1,i}) $ and $\mathbf{v}_j= (v_{1,j}, \cdots, v_{k-1,j})$, where $1 \le i \le n-d$ and $1 \le j \le d$.
Let $D_0=\left \{ (0, \mathbf{u}_i): 1 \le i \le n-d \right \}$, $D_1=\left \{ (1, \mathbf{v}_j): 1 \le j \le d \right \}$
and $B=\left \{ \mathbf{v}_j: 1 \le j \le d \right \}$. From Theorem \ref{thm:characterization}
it follows that $D_0 \dot \cup D_1$ is a vectorial cutting blocking set in the $k$-dimensional vector space $\gf (q)^k$.
We next claim that $B$ is an affine cutting blocking set in $\gf (q)^{k-1}$.
Let $H$ be any affine hyperplane in the affine space $\gf(q)^{k-1}$ defined by $a_1 x_1 + \cdots a_{k-1} x_{k-1}=b$,
where $(a_1, \cdots, a_{k-1}) \in \gf(q)^{k-1} \setminus \{\mathbf 0\}$ and $b\in \gf(q)$.
Let us denote by $\widetilde{H}$ the hyperplane $\left \{ (x_0, \cdots, x_{k-1}) \in \gf(q)^k: -b x_0 +a_1 x_1 + \cdots a_{k-1} x_{k-1}= 0 \right \}$.
Combining Proposition \ref{prop:cutting-fold}  with Theorem \ref{thm:characterization} we see that the set $\widetilde{H} \cap (D_0 \dot \cup D_1)$
is $(k-1)$-dimensional.
It is evident that $\mathrm{dim}\left ( \mathrm{Span} \left ( \widetilde{H} \cap D_0 \right ) \right ) \le  k-2$.
Consequently, the intersection of $\widetilde{H}$ and $D_1$ is not empty.
In particular, there exists a vector $(x_1, \cdots, x_{k-1}) \in B$ such
that $(1, x_1, \cdots, x_{k-1}) \in \widetilde{H} \cap D_1$
, that is, $a_1 x_1 + \cdots a_{k-1} x_{k-1}= b$. Thus $(x_1, \cdots, x_{k-1}) \in H \cap B$.
According to the above discussion, it follows that $B$ is an affine blocking set in $\gf(q)^{k-1}$.
The desired conclusion then follows from Theorem \ref{thm:Jamison77}.
\end{proof}

Combining Theorem \ref{thm:abs-bound} with Griesmer bound in (\ref{eq:Griesmer})  we get the following theorem, which improves
the results in \cite{ABN} and  \cite{LWC}.
\begin{theorem}
Let $k\ge 2$ and let $\C$ be a $q$-ary $[n,k]$ minimal  linear code. Then,
\[ n \ge (q-1)(k-1)+1 + \sum_{i=1}^{k-1} \left \lceil \frac{(q-1)(k-1)+1}{q^i} \right \rceil. \]
In particular, there does not exist a $q$-ary minimal linear code of length $n$ and dimension larger than $ \frac{n}{q}+1$.
\end{theorem}

The lower bound $(q-1)(k-1)+1 + \sum_{i=1}^{k-1} \left \lceil \frac{(q-1)(k-1)+1}{q^i} \right \rceil$
on the length $n$ of minimal $[n,k, d]_q$ linear codes is not very tight.
It would be nice if the following problem can be settled.

\begin{open}
Determine the minimum length of minimal $q$-ary linear codes with dimension $k$.
\end{open}

\section{New constructions of minimal linear codes}\label{sec:constr}

In this section, we present new primary and secondary constructions of minimal linear codes.
\subsection{Primary construction of minimal codes via unions of hyperplanes}
Let $H_{\mathbf{a}}$ denote
the hyperplane $\left \{ \mathbf x \in  \gf(q)^k : \langle \mathbf a , \mathbf x \rangle =0  \right \} $, where $\mathbf{a} \in \gf(q)^k$ is nonzero.
For any nonempty subset $S$ of $\gf(q)^k \setminus \{\mathbf 0\}$, let $D_{S}$ be the subset of $\gf(q)^k$ defined by
\begin{align}\label{D-S}
D_{S}=\left ( \cup_{\mathbf a \in S} H_{\mathbf a} \right ) \setminus \{\mathbf 0\}.
\end{align}
\begin{proposition}\label{prop:DS-cutting}
Let $S$ be a nonempty proper subset of $\gf(q)^k \setminus \{\mathbf 0\}$ where $k\ge 3$.
Then the set $D_S$ defined by (\ref{D-S}) is a vectorial cutting blocking set in $\gf(q)^k$,
if and only if, the dimension of $\mathrm{Span}(S)$ is at least $3$.

\end{proposition}
\begin{proof}
Assume that $\mathrm{dim}_{\gf(q)} \left ( \mathrm{Span}(S) \right ) =1$. Then
$D_S= H_{\mathbf a} \setminus \{\mathbf 0\}$ for some  $\mathbf a \in \gf(q)^k$.
Choose any $\mathbf a' \in \gf(q)^k \setminus \left \langle \mathbf a \right \rangle$, then
$H_{\mathbf a'+ \mathbf a} \cap D_S \subseteq H_{\mathbf a'}$.
Note that $H_{\mathbf a'+ \mathbf a}$ and  $H_{\mathbf a'}$ are two distinct hyperplanes, which shows
that $D_S$ is not a vectorial cutting blocking set in $\gf(q)^k$ from Definition \ref{def:cutting}.

For the case $\mathrm{dim}_{\gf(q)} \left ( \mathrm{Span}(S) \right ) =2$, there exist $\mathbf a , \mathbf a'\in \gf(q)^k$
such that \[D_S= \left ( \cup_{i=1}^{h} H_{\mathbf a + \lambda_i \mathbf a'} \right ) \cup \left ( H_{\mathbf a}  \cup H_{\mathbf a'} \right ),\]
where $\lambda_i \in \gf(q)^*$. Since $D_S$ is a proper subset of $\gf(q)^k$ and
$\gf(q)^k=  \left ( \cup_{\lambda \in \gf(q)^*} H_{\mathbf a + \lambda \mathbf a'} \right ) \cup \left ( H_{\mathbf a}  \cup H_{\mathbf a'} \right )$,
there exists a $\lambda \in \gf(q)^*$ such that $D_S$ does not contain the hyperplane  $H_{\mathbf a + \lambda \mathbf a'}$.
It is observed that $H_{\mathbf a + \lambda \mathbf a'} \cap H_{\mathbf a''}=H_{\mathbf a} \cap H_{\mathbf a'}$ for any $\mathbf a'' \in S$,
which implies that $H_{\mathbf a + \lambda \mathbf a'} \cap D_S \subseteq H_{\mathbf a} $.
By Definition \ref{def:cutting},  $D_S$ is not a vectorial cutting blocking set in $\gf(q)^k$.

Suppose that $\mathrm{dim}_{\gf(q)} \left ( \mathrm{Span}(S) \right ) \ge 3$.
 Let $H_{\mathbf a_1}$ and $H_{\mathbf a_2}$ be two distinct hyperplanes. Since the dimension of $\mathrm{Span}(S)$
is at least $3$, there exists $\mathbf a\in S$ such that $\{\mathbf a_1, \mathbf a_2, \mathbf a\}$
is $3$-dimensional. Then, there exists a solution $\mathbf x \in \gf(q)^k$ for the linear system
\[
\left\{
\begin{aligned}
\langle \mathbf a_1, \mathbf x \rangle & = 1,\\
\langle \mathbf a_2, \mathbf x \rangle & = 0,\\
\langle \mathbf a, \mathbf x \rangle & = 0.\\
\end{aligned}
\right.
\]
From $\mathbf x \in H_{\mathbf a} \subseteq D_S$,  one obtains $\mathbf x \in H_{\mathbf a_2} \cap  D_S$
and $\mathbf x \not\in H_{\mathbf a_1}$, that is, $H_{\mathbf a_2} \cap  D_S $ is not contained in the hyperplane
$H_{\mathbf a_1} $.
It follows that $D_S$ is a vectorial cutting blocking set in $\gf(q)^k$ from Definition \ref{def:cutting}.
This completes the proof.
\end{proof}

The following theorem is derived from Theorem \ref{thm:characterization} and Proposition \ref{prop:DS-cutting},
 which describes minimal codes constructed by unions of hyperplanes.
\begin{theorem}\label{thm:projC-D}
Let $S$ be a nonempty proper subset of $\gf(q)^k \setminus \{\mathbf 0\}$ where $k\ge 3$.
Let  $D_S$ be the set and $\C_{D_S}$ be the code defined by (\ref{D-S})  and (\ref{def:C-D}), respectively.  Then  $\C_{D_S}$  is
a minimal code of dimension $k$,
if and only if, the dimension of $\mathrm{Span}(S)$ is at least $3$.
\end{theorem}

As a result of Theorem \ref{thm:projC-D},  we have the following for  minimal projective codes.
\begin{corollary}
Let $S$ be a nonempty proper subset of $\gf(q)^k \setminus \{\mathbf 0\}$ where $k\ge 3$.
Let  $D_S$ be the set defined by (\ref{D-S})  and let $\overline{D}_S$  be any projection of $D_S$.  Then  $\C_{\overline{D}_S}$  is
a minimal projective  code of dimension $k$,
if and only if, the dimension of $\mathrm{Span}(S)$ is at least $3$.

\end{corollary}\label{lem:N(a,T)}
It is easily observed that the following lemma holds.
\begin{lemma}\label{lem:N(a,T)}
Let  $T$ be a subset of cardinality $t$ of $\{1,2,\cdots, k\}$ and $\mathbf a \in \gf(q)^k$.
Let   $N(\mathbf a, T)$  denote the number of solutions $\mathbf x = (x_1, \cdots, x_k)\in \gf(q)^k$  of the following  system of linear equations
$$
\begin{cases}
   x_j  =0, &\text{ for } j \in T,
   \cr  \langle \mathbf a , \mathbf x \rangle   =0.
\end{cases}
$$
Then
$$
N(\mathbf a, T)=
\begin{cases}
 q^{k-t},   & \text{ if} \quad \mathrm{Supp}(\mathbf a) \subseteq T,
 \cr   q^{k-t-1},  & \text{otherwise}.
\end{cases}
$$

\end{lemma}

\begin{theorem}\label{thm:x1xh=0}
Let $h$ and $k$ be two integers with $3\le h \le  k$ and
$D=\{(x_1, \cdots, x_k)\in \gf(q)^k \setminus \{0\}: x_1\cdots x_h=0\}$.
Then the code $\C_D$ in Equation (\ref{def:C-D}) is a
$\left [ (q^h-(q-1)^h)q^{k-h}-1,k,  w_{\min} \right ]$ minimal linear code
with the weight distribution in Table \ref{table:x1xh=0}, where $w_{\min}=q^{k-h} ((q-1)q^{h-1}-(q-1)^h)$. Furthermore,
\[ \frac{w_{\min}}{w_{\max}} \le \frac{q-1}{q},\]
if and only if $h\le 1+ \frac{1}{\log_{2}(\frac{q}{q-1})}$.
\begin{table}[htbp]
\centering
\caption{The weight distribution of the code  $\mathcal C_D$ of Theorem \ref{thm:x1xh=0}}
\label{table:x1xh=0}
\begin{tabular}{|c|c|}
  \hline
  Weight $w$ & No. of codewords $A_w$ \\
  \hline
  $0$& $1$\\
\hline
$(q-1) q^{k-h-1}\left (  q^h-(q-1)^h \right )$ & $q^k-q^h$\\
\hline
$(q-1) q^{k-h-1}\left (  q^h-(q-1)^h \right )$ & \\
$  +(-1)^{s} q^{k-h-1}(q-1)^{h-s+1}$ & $(q-1)^s \binom{h}{s}$\\
\text{ for } $s=1,2, \cdots, h$ & \\
\hline
\end{tabular}
\end{table}

\end{theorem}
\begin{proof}
By Theorem \ref{thm:projC-D}, $\mathcal C_D$ is a minimal code of dimension $k$.
Clearly, the length of $\mathcal C_D$ is $\# D= (q^h-(q-1)^h) q^{k-h}-1$.

Next, we consider the weight distribution of $\C_D$. Let $\mathbf c_{\mathbf a} = \left ( \langle \mathbf a , \mathbf x \rangle  \right )_{\mathbf x \in D}$ be a codeword in $\C_D$
corresponding to $\mathbf a \in \gf(q)^k \setminus \{\mathbf 0\}$.
Let $\mathrm{zr}(\mathbf c_{\mathbf a})$ be the number of solutions of the system of equations
$$
\left\{
  \begin{aligned}
    x_1 x_2 \cdots x_h  =0,  \\
    \langle \mathbf a , \mathbf x \rangle  =0.
  \end{aligned}
\right.
$$

Applying the inclusion-exclusion principle, we see that
\begin{align*}
\mathrm{zr}(\mathbf c_{\mathbf a})=& \sum_{i=1}^h (-1)^{i-1} \sum_{T\subseteq  \{1, \cdots, h\}, \# T =i} N(\mathbf a, T)\\
=&\sum_{T\subseteq  \{1, \cdots, h\}} (-1)^{\#T-1}N(\mathbf a, T)\\
=& \sum_{\mathrm{Supp}(\mathbf a) \not \subseteq  T\subseteq  \{1, \cdots, h\}} (-1)^{\#T-1}N(\mathbf a, T)\\
&+\sum_{ \mathrm{Supp}(\mathbf a)  \subseteq  T \subset \{1, \cdots, h\}} (-1)^{\#T-1}N(\mathbf a, T).
\end{align*}
where $N(\mathbf a, T)$
is defined as in Lemma \ref{lem:N(a,T)} if $T \neq  \emptyset$ and $N(\mathbf a, T)=0$ if $T= \emptyset$.
By Lemma \ref{lem:N(a,T)}, we deduce
\begin{align*}
\mathrm{zr}(\mathbf c_{\mathbf a})=& \sum_{\mathrm{Supp}(\mathbf a) \not \subseteq  T\subseteq  \{1, \cdots, h\}, T \neq \emptyset} (-1)^{\#T-1} q^{k-\#T-1}\\
&+\sum_{ \mathrm{Supp}(\mathbf a)  \subseteq  T \subset \{1, \cdots, h\}} (-1)^{\#T-1} q^{k-\#T}\\
=& \sum_{  T\subseteq  \{1, \cdots, h\}, T \neq \emptyset} (-1)^{\#T-1} q^{k-\#T-1}\\
&+(q-1)\sum_{ \mathrm{Supp}(\mathbf a)  \subseteq  T \subset \{1, \cdots, h\}} (-1)^{\#T-1} q^{k-\#T-1}\\
=&  \begin{cases}
  \sum_{i=1}^h (-1)^{i+1}  \binom{h}{i} q^{k-i-1},  & \text{if } \mathrm{Supp}(\mathbf a) \not \subseteq  \{1, \cdots, h\},\\
  \cr \sum_{i=1}^h (-1)^{i+1}  \binom{h}{i} q^{k-i-1}\\  +(q-1) \sum_{i=s}^{h} (-1)^{i+1}  \binom{h-s}{i-s}q^{k-i-1}, &\text{otherwise},
  \end{cases}
\end{align*}
where $s= \mathrm{wt}(\mathbf a)$. Then,
one obtains
\begin{align*}
\mathbf{wt}(\mathbf c_{\mathbf a})=& \# D- ( \mathrm{zr}(\mathbf c_{\mathbf a})-1) \nonumber \\
=& \begin{cases}
   (q^h-(q-1)^h) q^{k-h}-\sum_{i=1}^h (-1)^{i+1}  \binom{h}{i} q^{k-i-1}\\  -(q-1) \sum_{i=s}^{h} (-1)^{i+1}  \binom{h-s}{i-s}q^{k-i-1},
   & \text{if }  \mathrm{Supp}(\mathbf a)  \subseteq  \{1, \cdots, h\},\\
  \cr (q^h-(q-1)^h) q^{k-h}- \sum_{i=1}^h (-1)^{i+1}  \binom{h}{i} q^{k-i-1},  & \text{otherwise}.\\
  \end{cases}
  \end{align*}
By the identities $\sum_{i=s}^{h} (-1)^{i+1}  \binom{h-s}{i-s}q^{k-i-1}=(-1)^{s+1} q^{k-h-1}(q-1)^{h-s}$
and $\sum_{i=1}^h (-1)^{i+1} \binom{h}{i} q^{k-i+1}= q^{k-1}-q^{k-h-1}(q-1)^h $, one gets
\begin{align}\label{eq:wt-Dh}
\mathbf{wt}(\mathbf c_{\mathbf a})
= \begin{cases}
   (q-1) q^{k-h-1}\left (  q^h-(q-1)^h \right )\\  +(-1)^{s} q^{k-h-1}(q-1)^{h-s+1},
   & \text{if } \mathrm{Supp}(\mathbf a)  \subseteq  \{1, \cdots, h\},\\
  \cr (q-1) q^{k-h-1}\left (  q^h-(q-1)^h \right ),  & \text{otherwise},\\
  \end{cases}
  \end{align}
where $\mathbf a \neq \mathbf 0$ and $s= \mathrm{wt}(\mathbf a)$.
The weight distribution in Table \ref{table:x1xh=0} then follows from Equation (\ref{eq:wt-Dh}).

From the weight distribution of $\C_D$ in Table \ref{table:x1xh=0}, one gets
\begin{align*}
w_{\min}=(q-1)q^{k-h-1}(q^h-(q-1)^h-(q-1)^{h-1})
\end{align*}
and
\begin{align*}
w_{\max}=(q-1)q^{k-h-1}(q^h-(q-1)^h+(q-1)^{h-2}).
\end{align*}
Then
\begin{align}\label{eq:w/w}
\frac{w_{\min}}{w_{\max}}=\frac{q^{h-1}-(q-1)^{h-1}}{q^{h-1}-(q-2)(q-1)^{h-2}}.
\end{align}
Note that
\begin{align}\label{eq:A/B}
(q-1)B-qA=(q-1)^{h-1} \left (2-\left (\frac{q}{q-1} \right )^{h-1} \right ),
\end{align}
where $A=q^{h-1}-(q-1)^{h-1}$ and $B=q^{h-1}-(q-2)(q-1)^{h-2}$.
Combining Equations (\ref{eq:w/w}) and (\ref{eq:A/B}) gives that
$ \frac{w_{\min}}{w_{\max}} \le \frac{q-1}{q}$,
if and only if $h\le 1+ \frac{1}{\log_{2}(\frac{q}{q-1})}$.
This completes the proof.
\end{proof}
\begin{corollary}\label{thm:projx1xh=0}
Let $h$ and $k$ be two integers with $3\le h \le  k$ and
$D=\{(x_1, \cdots, x_k)\in \gf(q)^k \setminus \{0\}: x_1\cdots x_h=0\}$.
Let $\overline{D}$ be any projection of $D$.
Then the code $\C_{\overline{D}}$ in Equation (\ref{def:C-D}) is a minimal projective code
with parameters $\left [  \frac{q^k-1}{q-1} -q^{k-h}(q-1)^{h-1},k,  q^{k-h} (q^{h-1}-(q-1)^{h-1}) \right ]$
, whose weight distribution is listed in Table \ref{table:projx1xh=0}. Furthermore,
\[ \frac{w_{\min}}{w_{\max}} \le \frac{q-1}{q},\]
if and only if $h\le 1+ \frac{1}{\log_{2}(\frac{q}{q-1})}$.
\begin{table}[htbp]
\centering
\caption{The weight distribution of the code  $\mathcal C_{\overline{D}}$ of Corollary \ref{thm:projx1xh=0}}
\label{table:projx1xh=0}
\begin{tabular}{|c|c|}
  \hline
  Weight $w$ & No. of codewords $A_w$ \\
  \hline
  $0$& $1$\\
\hline
$ q^{k-h-1}\left (  q^h-(q-1)^h \right )$ & $q^k-q^h$\\
\hline
$ q^{k-h-1}\left (  q^h-(q-1)^h \right )$ & \\
$  +(-1)^{s} q^{k-h-1}(q-1)^{h-s}$ & $(q-1)^s \binom{h}{s}$\\
\text{ for } $s=1,2, \cdots, h$ & \\
\hline
\end{tabular}
\end{table}

\end{corollary}

\begin{corollary}\label{thm:2projx1xh=0}
Let $D=\{(x_1, \cdots, x_k)\in \gf(q)^k \setminus \{0\}: x_1 x_2 x_3=0\}$, where $q\ge 4$ and $k\ge 3$.
Let $\overline{D}$ be any projection of $D$.
Then the codes $\C_D$ and $\C_{\overline{D}}$ in Equation (\ref{def:C-D}) are minimal codes
with $ \frac{w_{\min}}{w_{\max}} < \frac{q-1}{q}.$
\end{corollary}
\begin{proof}
The conclusion follows from Theorem \ref{thm:x1xh=0}, Corollary \ref{thm:projx1xh=0} and the inequality  $1+\frac{1}{\log_{2}(\frac{q}{q-1})}\ge 1+\frac{1}{\log_{2}(\frac{4}{4-1})}\approx 3.41>3 $.
\end{proof}

The following numerical data is consistent with the conclusion of Theorem \ref{thm:x1xh=0}.
\begin{example}
Let $q=3$, $k=4$ and $h=3$. Then the set $\C_D$ in Theorem \ref{thm:x1xh=0} is a minimal code with parameters $[56,4,30]$ and weight enumerator
$1+6z^{30}+8z^{36}+54z^{38}+12z^{42}$.
Obviously, $h>1+\frac{1}{\log_{2}(\frac{q}{q-1})} \approx 2.71$ and $\frac{w_{\min}}{w_{\max}}=\frac{5}{7} >\frac{2}{3}$.
\end{example}

\begin{example}
Let $q=4$, $k=4$ and $h=3$. Then the set $\C_D$ in Theorem \ref{thm:x1xh=0} is a minimal code with parameters $[147,4,84]$ and weight enumerator
 $1+9z^{84}+27z^{108}+192z^{111}+27z^{120}$.
Obviously, $h<1+\frac{1}{\log_{2}(\frac{q}{q-1})} \approx 3.41$ and $\frac{w_{\min}}{w_{\max}}=\frac{7}{10} <\frac{3}{4}$.
\end{example}

\begin{lemma}\label{lem:lin-toric}
Let $h \ge 2$ and $b \in \gf(q)$. Let $N_{h,b}$ denote the cardinality of the set
$\{(x_1, \cdots, x_h)\in \gf(q)^h: x_1+\cdots x_h=b, x_i\neq 0 \text{ for } i=1, \cdots, h\}.$ Then
$$
N_{h,b}=
\left\{
  \begin{aligned}
     \frac{(q-1)^h+(-1)^h(q-1)}{q},   ~~\text{if}\quad   b=0, \\
   \frac{(q-1)^h-(-1)^h}{q},    ~~\text{if} \quad   b \neq 0.
  \end{aligned}
\right.
$$
\end{lemma}
\begin{proof}
By the definition of $N_{h,b}$, one has $N_{h,b}=N_{h,1}$ for $b \in \gf(q)^*$ and
\begin{align}\label{eq:0-1}
N_{h,0}+(q-1)N_{h,1}=(q-1)^h.
\end{align}
Plugging $N_{h,0}=\sum_{a \in \gf(q)^*} N_{h-1,a}=(q-1)N_{h-1,1}$ into Equation (\ref{eq:0-1}), we deduce that
\[ N_{h-1,1}+N_{h,1}=(q-1)^{h-1}. \]
That is
\[(N_{h,1}-q^{-1}(q-1)^h) + (N_{h-1,1} -q^{-1} (q-1)^{h-1})=0.\]
Then
\begin{align*}
N_{h,1}=&q^{-1}(q-1)^h+(-1)^{h-1}(N_{1,1}-q^{-1} (q-1)^{2-1})\\
=& \frac{(q-1)^h-(-1)^h}{q}.
\end{align*}
One gets
\begin{align*}
N_{h,0}=& (q-1)^h-(q-1) \frac{(q-1)^h-(-1)^h}{q}\\
=& \frac{(q-1)^h+(-1)^h(q-1)}{q}.
\end{align*}
This completes the proof.

\end{proof}

\begin{lemma}\label{lem:D1}
Let $h$ and $k$ be two integers with $2\le h \le  k$ and
$D=\{(x_1, \cdots, x_k)\in \gf(q)^k \setminus \{\mathbf 0\} : x_1\cdots x_h(x_1+\cdots +x_h)=0\}$.
Then, $\# D= q^{k-h-1}\left ( q^{h+1}-(q-1)^{h+1}+(-1)^{h}(q-1) \right )-1$.
\end{lemma}

\begin{proof}
By the definition of the set $D$, we have the following decomposition of $D$ as  Cartesian product
\begin{align*}
D=\left ( D_1\times \gf(q)^{k-h} \right ) \setminus \{\mathbf 0\},
\end{align*}
where $D_1$ is defined as
\begin{align}\label{eq:D1}
D_1=\{(x_1, \cdots, x_h)\in \gf(q)^h : x_1\cdots x_h(x_1+\cdots +x_h)=0\}.
\end{align}
Then
\begin{align*}
\# D_1=& \# \{(x_1, \cdots, x_h)\in \gf(q)^h : x_1\cdots x_h=0\}\\
&+\# \{(x_1, \cdots, x_h)\in \gf(q)^h : x_1+\cdots +x_h=0\}\\
& - \# \{(x_1, \cdots, x_h)\in \gf(q)^h : x_1\cdots x_h=x_1+\cdots +x_h=0\}\\
=& q^h-(q-1)^h+q^{h-1}\\
&- \# \{(x_1, \cdots, x_h)\in \gf(q)^h : x_1\cdots x_h=x_1+\cdots +x_h=0\}\\
=&(q+1)q^{h-1}-(q-1)^{h}\\
&- \# \{(x_1, \cdots, x_h)\in \gf(q)^h :x_1+\cdots +x_h=0\}\\
&+  \# \{(x_1, \cdots, x_h)\in \left ( \gf(q)^* \right )^h :x_1+\cdots +x_h=0\}.
\end{align*}
From Lemma \ref{lem:lin-toric}, one has
\begin{align*}
\# D_1= q^{h}-(q-1)^{h}+\frac{(q-1)^h+(-1)^h(q-1)}{q}.
\end{align*}
Thus, $\# D=q^{k-h-1}\left ( q^{h+1}-(q-1)^{h+1}+(-1)^{h}(q-1) \right )-1$.
This completes the proof.

\end{proof}

\begin{theorem}\label{thm:x1xh(+)=0}
Let $h$ and $k$ be two integers with $3\le h \le  k$ and
$D=\{(x_1, \cdots, x_k)\in \gf(q)^k \setminus \{0\}: x_1\cdots x_h(x_1+\cdots +x_h)=0\}$.
Then the set $\C_D$ in Equation (\ref{def:C-D}) is an
$\left [ n,k,  w_{\min} \right ]$ minimal linear code, where
$n=q^{k-h-1}\left ( q^{h+1}-(q-1)^{h+1}+(-1)^{h}(q-1) \right )-1$
 and \[w_{\min}= q^{k-h-1}
 \left  ((q-1)q^{h}-(q-1)^{h+1}+(-1)^h(q-1) \right ).\] Furthermore,
\[ \frac{w_{\min}}{w_{\max}} \le \frac{q-1}{q},\]
provided that $2(q-1)^h-q^h+(-1)^h(q-2) \ge 0$.

\end{theorem}

\begin{proof}
By Theorem \ref{thm:projC-D}, $\mathcal C_D$ is a minimal code of dimension $k$.
The length of the code $\C_D$ follows from Lemma \ref{lem:D1}.
Let $\mathbf c_{\mathbf a} = \left ( \langle \mathbf a , \mathbf x \rangle  \right )_{\mathbf x \in D}$ be a codeword in $\C_D$
corresponding to $\mathbf a \in \gf(q)^k \setminus \{\mathbf 0\}$.
Note that
\begin{align*}
w_{\min}=&n-\# \{ \mathbf x \in  D:   \langle \mathbf a , \mathbf x \rangle =0\}\\
\ge & n+1- \# \{ \mathbf x \in  \gf(q)^k:   \langle \mathbf a , \mathbf x \rangle =0\}\\
=& n+1-q^{k-1}.
\end{align*}
The desired value of $w_{\min}$ then follows from $\mathrm{wt}(\mathbf c)= n+1-q^{k-1}$,
 where $\mathbf c= \left ( x_1  \right )_{(x_1,\cdots, x_k) \in D}$.

 Next, consider the weight $w_2$ of the codeword $\mathbf c= \left ( ax_{h-1}-x_h  \right )_{(x_1,\cdots, x_k) \in D}$, where $a\neq 0 \text{ or } -1$.
 Then
 \begin{align*}
 w_2= & n+1-\# \{ (x_1, \cdots, x_k)\in D: a x_{h-1}-x_h=0 \}\\
 =& n+1- q^{k-h} \# D_2,
 \end{align*}
where $D_2$ denotes the set of solutions $(x_1, \cdots, x_h)\in \gf(q)^h$ of the system of linear equations
\begin{align*}
\begin{cases}
 x_1 \cdots x_h (x_1+ \cdots +x_h)=0
 \cr a x_{h-1}-x_h=0
 \end{cases}.
\end{align*}
Thus, $\# D_2$ equals the number of solutions $(x_1, \cdots, x_{h-1})\in \gf(q)^{h-1}$ of the equation
\[x_1 \cdots x_{h-1} (x_1+ \cdots +x_{h-2}+(1+a)x_{h-1})=0.\]
Using the linear transformation $\begin{cases}
x_j'=x_j, & 1 \le j \le h-2
\cr x_{j}'=(1+a) x_{j},  & j=h-1
\end{cases}
$ over $\gf(q)^{h-1}$, we deduce that the equations $x_1 \cdots x_{h-1} (x_1+ \cdots+(1+a)x_{h-1})=0$ and
$x_1 \cdots x_{h-1} (x_1+ \cdots +x_{h-1})=0$ have the same number  of solutions over $\gf(q)^{h-1}$.
By Lemma \ref{lem:D1}, $\# D2=q^{k-h-1}  (  q^{h}-(q-1)^k +(-1)^{h-1}(q-1))$, which implies
$w_2=q^{k-h-1} \left (  (q-1)q^{h}-(q-2)(q-1)^h +2(-1)^{h}(q-1)\right )$.
It's easily checked that
\[(q-1)w_2-qw_{\min}=(q-1)q^{k-h-1}\left (2(q-1)^h-q^h+(-1)^h(q-2)\right ).\]
Hence, if $2(q-1)^h-q^h+(-1)^h(q-2)\ge 0$, then $\frac{w_{\min}}{w_{\max}}\le \frac{w_{\min}}{w_{2}} \le \frac{q-1}{q}$.
This completes the proof.

\end{proof}

Let $h\ge 3$ be an integer. Theorem \ref{thm:x1xh(+)=0} shows that  there exists a constant $q_0$
such that the code $C_D$ in Theorem \ref{thm:x1xh(+)=0} is a minimal code with $\frac{w_{\min}}{w_{\max}} \le \frac{q-1}{q}$
for any power $q\ge q_0$ of prime. Finally, we settle the weight distribution of the codes
in Theorem \ref{thm:x1xh(+)=0} for the case $h=3$.

\begin{theorem}\label{thm:x1x3(+)=0}
Let $k$ be an integer with $  k\ge 3$ and
$D=\{(x_1, \cdots, x_k)\in \gf(q)^k \setminus \{0\}: x_1 x_2 x_3(x_1+x_2+x_3)=0\}$.
Then the set $\C_D$ in Equation (\ref{def:C-D}) is an
$\left [ n,k,  w_{\min} \right ]$ minimal linear code
with the weight distribution in Table \ref{table:x1x3(+)=0}, where $n=4q^{k-1}-6q^{k-2}+3q^{k-3}-1$
and  $w_{\min}=3q^{k-1}-6q^{k-2}+3q^{k-3}$. Furthermore,
\[ \frac{w_{\min}}{w_{\max}} \le \frac{q-1}{q},\]
if and only if $q\ge 4$.
\begin{table}[htbp]
\centering
\caption{The weight distribution of the code  $\mathcal C_D$ of Theorem \ref{thm:x1x3(+)=0}}
\label{table:x1x3(+)=0}
\begin{tabular}{|c|c|}
  \hline
  Weight $w$ & No. of codewords $A_w$ \\
  \hline
  $0$& $1$\\
\hline
$3q^{k-1}-6q^{k-2}+3q^{k-3}$ & $4(q-1)$\\
\hline
$4q^{k-1}-10q^{k-2}+6q^{k-3}$ & $(q-1)(q-2)(q-3)$\\
\hline
$4q^{k-1}-10q^{k-2}+9q^{k-3}-3q^{k-4}$ & $q^k-q^3$\\
\hline
$4q^{k-1}-9q^{k-2}+5q^{k-3}$ & $6(q-1)(q-2)$\\
\hline
$4q^{k-1}-8q^{k-2}+4q^{k-3}$ & $3(q-1)$\\
\hline
\end{tabular}
\end{table}

\end{theorem}

\begin{proof}
By Theorem \ref{thm:x1xh(+)=0}, $\mathcal C_D$ is a minimal code of dimension $k$
with length $n=4q^{k-1}-6q^{k-2}+3q^{k-3}-1$ and minimum weight $w_{\min}=3q^{k-1}-6q^{k-2}+3q^{k-3}$.

Next, we consider the weight distribution of $\C_D$. Let $\mathbf c_{\mathbf a} = \left ( \langle \mathbf a , \mathbf x \rangle  \right )_{\mathbf x \in D}$ be a codeword in $\C_D$
corresponding to $\mathbf a \in \gf(q)^k \setminus \{\mathbf 0\}$.
Let $\mathrm{zr}(\mathbf c_{\mathbf a})$ be the number of solutions of the system of equations
$$
\left\{
  \begin{array}{l}
    x_1 x_2 x_3(x_1+ x_2+ x_3)=0 \\
    \langle \mathbf a , \mathbf x \rangle =0
  \end{array}
\right.
$$

Let   $N'(\mathbf a, T)$  denote the number of solutions $\mathbf x = (x_1, \cdots, x_k)\in \gf(q)^k$  of the following  system of linear equations
$$
\left\{
  \begin{array}{l}
    x_j =0, \text{ for } j \in T, \\
    x_1+x_2+x_3=0,\\
     \langle \mathbf a , \mathbf x \rangle  =0
  \end{array}
\right.
$$
Employing the inclusion-exclusion principle, we have
\begin{align}\label{eq:zr(+)}
\mathrm{zr}(\mathbf c_{\mathbf a})=& \sum_{i=0}^3 (-1)^{i-1} \sum_{T\subseteq  \{1, 2, 3\}, \# T =i} (N(\mathbf a, T)-N'(\mathbf a, T))\nonumber \\
=&\sum_{T\subseteq  \{1, 2, 3\}} (-1)^{\#T-1}(N(\mathbf a, T)-N'(\mathbf a, T)),
\end{align}
where $N(\mathbf a, T)$
is defined as in Lemma \ref{lem:N(a,T)} if $T \neq  \emptyset$ and $N(\mathbf a, T)=0$ if $T= \emptyset$.
Then, we deduce
\begin{align*}
\mathrm{zr}(\mathbf c_{\mathbf a})
=&  \begin{cases}
  q^{k-1},  & \text{if } \langle \mathbf a , \mathbf x \rangle= a x_i , a(x_1+x_2+x_3),
  \cr 2 q^{k-2}-q^{k-3},& \text{if } \langle \mathbf a , \mathbf x \rangle= a (x_i+x_j),
  \cr 3q^{k-2}-2q^{k-3}, & \text{if } \langle \mathbf a , \mathbf x \rangle= a x_i+bx_j, a x_i+ax_j+b x_k,
  \cr 4q^{k-2}-3q^{k-3}, & \text{if } \langle \mathbf a , \mathbf x \rangle= a x_1+bx_2+c x_3,
    \cr 4q^{k-2}-6q^{k-3}+3q^{k-4}, &\text{if }  \mathrm{Supp}(\mathbf a)  \not \subseteq  \{1, 2, 3\},
  \end{cases}
\end{align*}
where $i,j,k$ are pairwise distinct integers in $\{1,2,3\}$  and $a, b ,c \in \gf(q)^*$ are pairwise distinct. Then,
one obtains
\begin{align}\label{eq:wt-Dh+}
\mathbf{wt}(\mathbf c_{\mathbf a})=& n- ( \mathrm{zr}(\mathbf c_{\mathbf a})-1) \nonumber \\
=&  \begin{cases}
  3q^{k-1}-6q^{k-2}+3q^{k-3},  &\text{if }  \langle \mathbf a , \mathbf x \rangle= a x_i , a(x_1+x_2+x_3),
  \cr 4q^{k-1}-8q^{k-2}+4q^{k-3} ,& \text{if } \langle \mathbf a , \mathbf x \rangle= a (x_i+x_j),
  \cr 4q^{k-1}-9q^{k-2}+5q^{k-3} , & \text{if } \langle \mathbf a , \mathbf x \rangle= a x_i+bx_j, a x_i+ax_j+b x_k,
  \cr 4q^{k-1}-10q^{k-2}+6q^{k-3}, & \text{if } \langle \mathbf a , \mathbf x \rangle= a x_1+bx_2+c x_3,
    \cr 4q^{k-1}-10q^{k-2}+9q^{k-3}-3q^{k-4}, &\text{if }  \mathrm{Supp}(\mathbf a)  \not \subseteq  \{1, 2, 3\},
  \end{cases}
  \end{align}
  where $i,j,k$ are pairwise distinct integers in $\{1,2,3\}$  and $a, b ,c \in \gf(q)^*$ are pairwise distinct.
The weight distribution in Table \ref{table:x1x3(+)=0} then follows from Equation (\ref{eq:wt-Dh+}).

From the weight distribution of $\C_D$ in Table \ref{table:x1x3(+)=0}, one gets
\begin{align*}
w_{\max}=4q^{k-1}-8q^{k-2}+4q^{k-3}.
\end{align*}
Then
\begin{align}\label{eq:w/w(+)}
\frac{w_{\min}}{w_{\max}}=\frac{3}{4}.
\end{align}
Equation (\ref{eq:w/w(+)}) implies that
$ \frac{w_{\min}}{w_{\max}} \le \frac{q-1}{q}$,
if and only if $q\ge 4$.
This completes the proof.
\end{proof}

The following numerical data is consistent with the conclusion of Theorem \ref{thm:x1x3(+)=0}.
\begin{example}
Let $q=3$, $k=4$ and $h=3$. Then the set $\C_D$ in Theorem \ref{thm:x1x3(+)=0} is a minimal code with parameters $[62,4,36]$ and weight enumerator
$1+8z^{36}+66z^{42}+6z^{48}$.
Obviously, $\frac{w_{\min}}{w_{\max}}=\frac{3}{4} >\frac{2}{3}$.
\end{example}

\begin{example}
Let $q=4$, $k=4$ and $h=3$. Then the set $\C_D$ in Theorem \ref{thm:x1x3(+)=0} is a minimal code with parameters $[171,4,108]$ and weight enumerator
 $1+12z^{108}+6z^{120}+192z^{129}+36z^{132}+9z^{144}$.
Obviously,  $\frac{w_{\min}}{w_{\max}}=\frac{3}{4}$.
\end{example}

The following open problems would be challenging.
\begin{open}
Determine the weight distribution of the minimal code $\C_D$ for the case that
\[ D:=\left \{(x_1, \cdots, x_k)\in \gf(q)^k \setminus \{0\}:\left (\sum_{i=1}^h x_i \right ) \prod_{i=1}^h x_i =0 \right \},\]
where $4\le h\le  k$.
\end{open}

\begin{open}
Determine the parameters of the linear code $\C_D$ for the case that
\[ D:=\left \{(x_1, \cdots, x_k)\in \gf(q)^k \setminus \{0\}:\prod_{1\le i <j \le h}(x_i+x_j)=0 \right \},\]
where $3\le h\le  k$.
\end{open}

\begin{open}
Determine the parameters of the minimal code $\C_D$ for the case that
\[ D:=\left \{(x_1, \cdots, x_k)\in \gf(q)^k \setminus \{0\}: \prod_{i=1}^h x_i \prod_{1\le i <j \le h}(x_i+x_j)=0 \right \},\]
where $3\le h\le  k$.
\end{open}

\begin{open}
Determine the parameters of the minimal code $\C_D$ for the case that
\[ D:=\left \{(x_1, \cdots, x_k)\in \gf(q)^k \setminus \{0\}:\left (\sum_{i=1}^h x_i \right ) \prod_{i=1}^h x_i \prod_{1\le i <j \le h}(x_i+x_j)=0\right \},\]
where $3\le h\le  k$.
\end{open}

Other good minimal codes may be produced by choosing an appropriate union of hyperplanes.

\subsection{Secondary constructions of minimal codes}
We introduce now a secondary construction of minimal linear codes,
 which  allows constructing minimal codes of dimension $(k+1)$ from $k$-dimensional minimal codes.

\begin{lemma}\label{lem:cutting-lineq}
Let $k \ge 2$. Let ${D}$  be a vectorial  cutting blocking set in $\gf(q)^k$ such that $D=a \cdot  D$ for any $a \in \gf(q)^*$.
Then, for any $\mathbf a \in \gf(q)^*$ and $c\in \gf(q)$, there exists $\mathbf x\in D$ such that $\langle \mathbf a , \mathbf x \rangle =c$.
\end{lemma}
\begin{proof}
Suppose that there exists $\mathbf a \in \gf(q)^*$ and $c\in \gf(q)$ such that $\langle \mathbf a , \mathbf x \rangle \neq c$ for any  $\mathbf x\in D$.
Let $H_{\mathbf a} $ be the hyperplane corresponding to $\mathbf a$.
If $c=0$, then $H_{\mathbf a} \cap D =\emptyset$, which is contrary to Proposition \ref{prop:cutting-fold}.
If $c\neq 0$, then $D\subseteqq H_{\mathbf a}$. Thus, for any other  hyperplane $H'$, $H' \cup D\subseteqq H_{\mathbf a}$, which is contrary to
the definition of cutting blocking set.
This completes the proof.
\end{proof}

The following follows from the definition of cutting blocking sets directly.
\begin{lemma}\label{lem:cutting-lineq*}
Let $k \ge 2$. Let ${D}$ be a vectorial  cutting blocking set in $\gf(q)^k$. Let $\mathbf a_1, \mathbf a_2 \in \gf(q)^k$
such that they are linearly independent.
\begin{enumerate}
\item  There exist $\mathbf x, \mathbf x' \in D$ such that $\langle \mathbf a_1 , \mathbf x \rangle =0$
and $\langle \mathbf a_1 , \mathbf x' \rangle \neq 0$.
\item There exists $\mathbf x \in D$ such that $\langle \mathbf a_1 , \mathbf x \rangle  \neq 0$
and $\langle \mathbf a_2 , \mathbf x \rangle \neq 0$.
\end{enumerate}
\end{lemma}

We now present the secondary construction of minimal linear codes in the next theorem.
\begin{theorem}\label{thm:D1-D2}
Let $k \ge 2$. Let ${D_1}$ and ${D_2}$ be two vectorial  cutting blocking sets in $\gf(q)^k$ such that $D_1=a \cdot  D_1$ for any $a\in \gf(q)^*$.
Define the subset of $\gf(q)^{k+1}$ by
\begin{align}\label{eq:D1-D2}
\widetilde{[D_1, D_2]}:=\left \{ (\mathbf x,1)\in \gf(q)^{k+1}: \mathbf x \in D_1 \right \} \bigcup \left \{ ( \mathbf x,0)\in \gf(q)^{k+1}: \mathbf x \in D_2 \right \}.
\end{align}
Then, $\widetilde{[D_1, D_2]}$ is a vectorial cutting blocking set in $\gf(q)^{k+1}$.
In particular, $\C_{\widetilde{[D_1, D_2]}}$ is a minimal code of length $(\# D_1+ \#D_2)$ and dimension $(k+1)$.
\end{theorem}

\begin{proof}
By Theorem \ref{thm:characterization}, the definition of cutting blocking sets and the assumption that $D_1=a \cdot  D_1$ for any $a\in \gf(q)^*$, we only need to prove that
for any $(\mathbf a_i, b_i)\in \gf(q)^k \times \gf(q)$  $(i=1,2)$, where $(\mathbf a_1, b_1)$ and
$(\mathbf a_2, b_2)$ are linearly   independent over $\gf(q)$, there exists $(\mathbf x, y)\in D_1 \times \gf(q)^* \bigcup D_2 \times \{\mathbf 0\}$
such that
\begin{align}\label{eq:sep}
\begin{cases}
\langle \mathbf a_1, \mathbf x \rangle +b_1y \neq 0,
\cr \langle \mathbf a_2, \mathbf x \rangle +b_2y = 0.
\end{cases}
\end{align}
The proof is carried out by considering the  following three cases.

If $\mathbf a_1$ and $\mathbf a_2$ are linearly independent, by the definition of cutting blocking sets,
there exists an $\mathbf x  \in D_2$ such that $\langle \mathbf a_1, \mathbf x \rangle  \neq 0$ and $\langle \mathbf a_2, \mathbf x \rangle  = 0$.
Then, $(\mathbf x , 0) \in D_2 \times \{\mathbf 0\}$ satisfies (\ref{eq:sep}).

If $\mathbf a_1 = \lambda \mathbf a_2$ and $\mathbf a_2 \neq \mathbf 0$, by Lemma \ref{lem:cutting-lineq}, there
exists an $\mathbf x \in D_1$ such that $\langle \mathbf a_2, \mathbf x \rangle +b_2 = 0$.
Then $(\mathbf x ,1) \in D_1 \times \gf(q)^*$ satisfies (\ref{eq:sep}) from the fact  $b_1\neq \lambda b_2$.

If $\mathbf a_2 = \mathbf 0$, then $\mathbf a_1 \neq 0$. By Lemma \ref{lem:cutting-lineq*},
there
exists an $\mathbf x \in D_2$ such that $\langle \mathbf a_1, \mathbf x \rangle \neq 0$.
Thus, $(\mathbf x , 0) \in D_2 \times \{\mathbf 0\}$ satisfies (\ref{eq:sep}).

This completes the proof.

\end{proof}

For any cutting blocking set $D$, if $D$ does not satisfy the condition $D=aD$ for any $a\in \gf(q)^*$,
then the cutting blocking set $D'$ satisfies the previous condition, where $D'=\{a \mathbf x: a\in \gf(q), \mathbf x \in D\}$.
Let $D_1$ and $D_2$ be two vectorial  cutting blocking sets in $\gf(q)$. Define
\begin{align*}
\widetilde{\widetilde{[D_1, D_2]}}:=\left \{ (\mathbf x, y_{\mathbf x})\in \gf(q)^{k+1}: \mathbf x \in D_1 \right \} \bigcup \left \{ ( \mathbf x,0)\in \gf(q)^{k+1}: \mathbf x \in D_2 \right \},
\end{align*}
where $y_{\mathbf x} \in \gf(q)^*$ and $y_{a\mathbf x} =y_{\mathbf x}$ for any $a\in \gf(q)^*$.
Then, the code $\C_{\widetilde{\widetilde{[D_1, D_2]}}}$ is also a minimal code
as $\C_{\widetilde{\widetilde{[D_1, D_2]}}}$  and $\C_{{\widetilde{[D_1, D_2]}}}$ are equivalent up to a
monomial transformation.

\begin{corollary}
Let ${D}$ be a vectorial  cutting blocking set in $\gf(q)^k$ such that $D=a \cdot  D$ for any $a\in \gf(q)^*$.
Then $\C_{\widetilde{[D, D]}}$ is a minimal code of length $2\# D$ and dimension $(k+1)$, where $\widetilde{[D, D]}$ is given by (\ref{eq:D1-D2}).
Furthermore, if $\mathcal C_D$ satisfies the condition $\frac{w_{\min}}{w_{\max}} \le \frac{q-1}{q}$, then so does $\C_{\widetilde{[D, D]}}$.
\end{corollary}
\begin{proof}
The conclusions follow from Theorem \ref{thm:D1-D2} and the fact that
if $\mathbf c \in \C_D$ then $(\mathbf c, \mathbf c) \in \C_{\widetilde{[D, D]}}$.
\end{proof}

Similarly, one can prove the following corollary.
\begin{corollary}
Let ${D}$ be a vectorial  cutting blocking set in $\gf(q)^k\setminus \{\mathbf 0\}$ such that $D=a \cdot  D$ for any $a\in \gf(q)^*$.
Let $\overline{D}$ be any projection of $D$.
Then $\C_{\widetilde{[D, \overline{D}]}}$ is a minimal projective code of length $\frac{q\# D}{q-1}$ and dimension $(k+1)$,
where $\widetilde{[D, \overline{D}]}$ is given by (\ref{eq:D1-D2}). Furthermore, if $\mathcal C_D$ satisfies  the condition
$\frac{w_{\min}}{w_{\max}} \le \frac{q-1}{q}$, then so does $\C_{\widetilde{[D, \overline{D}]}}$.
\end{corollary}

The following is a consequence of Proposition \ref{prop:DS-cutting} and Theorem  \ref{thm:D1-D2}.
\begin{theorem}\label{thm:DS-DS}
Let $S_1$ and $S_2$ be two nonempty proper subsets of $\gf(q)^k \setminus \{\mathbf 0\}$ such that the dimension of $\mathrm{dim}_{\gf(q)}(\mathrm{Span}(S_i))$ ($i=1,2$) is at least $3$.
Let  $D_{S_i}$ be the set  defined by (\ref{D-S}).  Then   $\C_{\widetilde{[D_{S_1}, {D_{S_2}}]}}$ is a minimal linear   code of dimension $(k+1)$,
where $\widetilde{[D_{S_1}, D_{S_2}]}$ is given by (\ref{eq:D1-D2}).
\end{theorem}

\begin{example}
Let $q=4$, $k=5$ and $D=\{(x_1, \cdots, x_5)\in \gf(q)^5 \setminus \{\mathbf 0\} : x_1 x_2 x_3=0\}$. Then the set $C_{\widetilde{[D, D]}}$ in Theorem \ref{thm:DS-DS} is a minimal code with parameters $[1182,6,591]$ and maximum weight $w_{\max}=960$.
Obviously,  $\frac{w_{\min}}{w_{\max}}=\frac{197}{320}<\frac{3}{4}$.
\end{example}

It would be interesting to settle the following problem.
\begin{open}
Determine the weight distribution of the minimal code $\C_{\widetilde{[D, {D}]}}$ for the case that
\[ D:=\left \{(x_1, \cdots, x_k)\in \gf(q)^k \setminus \{0\}: \prod_{i=1}^h x_i =0 \right \},\]
where $3 \le h\le  k$.
\end{open}

\begin{theorem}\label{thm:mc-wt}
Let $h$, $\ell$ and $k$ be integers such that $h\ge 2$ and $\ell  \le (k-2)$.  Let $D_{\le h}$ and $D_{\ge \ell}$  be the subsets of $\gf(q)^k$ defined by
\begin{align*}
D_{\le h}=\left \{ \mathbf x \in  \gf(q)^k : 1\le \mathrm{wt}(\mathbf x)\le h  \right  \},
\end{align*}
and
\begin{align*}
D_{\ge \ell}=\left \{ \mathbf x \in  \gf(q)^k : \ell \le \mathrm{wt}(\mathbf x)\le k  \right  \}.
\end{align*}
Then the codes   $\C_{\widetilde{[D_{\le h}, {D_{\le h}}]}}$, $\C_{\widetilde{[D_{\ge \ell}, {D_{\ge \ell}}]}}$  and
$\C_{\widetilde{[D_{\le h}, {D_{\ge \ell}}]}}$  are  minimal linear   codes of dimension $(k+1)$,
\end{theorem}
\begin{proof}
From  \cite{BB19} and \cite{MQRT}, $D_{\le h}$ and $D_{\ge \ell}$ are cutting blocking sets.
The desired conclusion then follows from Theorem \ref{thm:D1-D2}.
\end{proof}

As a special case, the minimal  codes $\C_{\widetilde{[D_{\le h}, {D_{\ge (h+1)}}]}}$ have appeared in \cite{DHZI,DHZF} and
 \cite{BBI}. Thus, Theorem \ref{thm:mc-wt} generalizes some previously known constructions of minimal codes.

\begin{example}
Let $q=3$, $k=6$. Then the set $\C_{\widetilde{[D_{\le 2}, {D_{\le 2}}]}}$ in Theorem \ref{thm:mc-wt} is a minimal code with parameters $[144,7,44]$ and maximum weight
$w_{\max}=104$.
Obviously,  $\frac{w_{\min}}{w_{\max}}=\frac{11}{26}<\frac{2}{3}$.
\end{example}

The following open problem would be interesting.
\begin{open}
Determine the weight distributions of the minimal codes $\C_{\widetilde{[D_{\le h}, {D_{\le h}}]}}$, $\C_{\widetilde{[D_{\ge \ell}, {D_{\ge \ell}}]}}$  and
$\C_{\widetilde{[D_{\le h}, {D_{\ge \ell}}]}}$ in Theorem \ref{thm:mc-wt}.
\end{open}

The following theorem presents a general approach to constructing minimal codes via cutting blocking sets and affine blocking sets.
\begin{theorem}\label{thm:cutting-affine blocking sets}
Let $k \ge 3$. Let $D_1$ be a subset of $\gf(q)^{k-1}$ and let ${D_2}$ be a vectorial cutting blocking sets in $\gf(q)^{k-1}$.
Define the subset of $\gf(q)^{k}$ by
\begin{align}\label{eq:D1-D2}
\widetilde{[D_1, D_2]}:=\left \{ (\mathbf x,1)\in \gf(q)^{k}: \mathbf x \in D_1 \right \} \bigcup \left \{ ( \mathbf x,0)\in \gf(q)^{k}: \mathbf x \in D_2 \right \}.
\end{align}
Then, the code $\C_{\widetilde{[D_1, D_2]}}$ is a minimal code if and only if the set $D_1$ is an affine blocking set in $\gf(q)^{k-1}$.
\end{theorem}
\begin{proof}
We first suppose that $\C_{\widetilde{[D_1, D_2]}}$ is a minimal code.
In the same manner as in the proof of Theorem \ref{thm:abs-bound} we can see that
$D_1$ is an affine blocking set.

Conversely, suppose that $D_1$ is an affine blocking set in $\gf(q)^{k-1}$.
Analysis similar to that in the proof of Theorem \ref{thm:D1-D2} shows
that $\C_{\widetilde{[D_1, D_2]}}$ is a minimal code.
\end{proof}

\begin{corollary}\label{cor:cutting+affine}
Let $k \ge 3$. Let $D_1=\left \{ a \mathbf{e}_i: a\in \gf(q), 1 \le i \le k-1 \right \}$, where $\mathbf{e}_1, \cdots, \mathbf{e}_{k-1}$
is a basis for $\gf(q)^{k-1}$ over $\gf(q)$. Let ${D_2}$ be a vectorial cutting blocking sets in $\gf(q)^{k-1}$.
Define the subset of $\gf(q)^{k}$ by
\begin{align}\label{eq:D1-D2}
\widetilde{[D_1, D_2]}:=\left \{ (\mathbf x,1)\in \gf(q)^{k}: \mathbf x \in D_1 \right \} \bigcup \left \{ ( \mathbf x,0)\in \gf(q)^{k}: \mathbf x \in D_2 \right \}.
\end{align}
Then, the code $\C_{\widetilde{[D_1, D_2]}}$ is a minimal code with parameters $[\# D_2 + (k-1)(q-1)+1, k, (k-1)(q-1)+1]$.
\end{corollary}
\begin{proof}
Let $H$ be an affine hyperplane given by $a_1x_1 + \cdots + a_{k-1} x_{k-1}=b$.
￼Since any nonzero linear functional must be nonzero on at least one basis vector, it is clear that
multiplying that basis vector by an appropriate scalar yields a solution in $D_1$ to
$a_1x_1 + \cdots + a_{k-1} x_{k-1}=b$.
Hence, $D_1$ is an affine blocking set containing $(k-1)(q-1)+1$ vectors.
The desired conclusion then follows from Theorems \ref{thm:abs-bound} and \ref{thm:cutting-affine blocking sets}.
\end{proof}

\begin{remark}
Corollary \ref{cor:cutting+affine} presents a generic construction
of minimal codes with minimum distance achieving the lower bound in Theorem \ref{thm:abs-bound}.
In particular, the lower bound in Theorem \ref{thm:abs-bound} is tight.
\end{remark}

\begin{remark}
Let $\mathbf{e}_i$ be the vector in $\gf(q)^{k-1}$ with all elements equal to zero,
except $i$th element which is equal to one.
Denote $D_1=\left \{ a \mathbf{e}_i: a\in \gf(q), 1 \le i \le k-1 \right \}$.
Let $D_2$ be the set given by
 \[ D_2=\left \{ \mathbf{e}_i + a \mathbf{e}_j: a \in \gf(q)^{*}, 1 \le i <j \le k-1  \right \} \cup \left  (D_1 \setminus \{\mathbf 0\} \right ). \]
 Then the code $\C_{\widetilde{[D_1, D_2]}}$ is a minimal code. It is easily seen that
 $\C_{\widetilde{[D_1, D_2]}}$ is equivalent to the code in Theorem 5.4 of \cite{ABN}.
 Thus Corollary \ref{cor:cutting+affine} gives a new interpretation of this code.
\end{remark}

\section{Summary and concluding remarks} \label{sec:conc}

The main contributions of this paper are the following.
\begin{itemize}
\item A link between minimal linear codes and blocking sets
 was established and documented in Theorem \ref{thm:characterization}, which
  says that projective minimal codes and  cutting blocking sets are identical objects.
  Adopting  this geometric perspective a tight lower bound on the minimum distance of minimal codes is yielded, which
  confirms a recent conjecture by Alfarano, Borello and Neri.
   The link also played an important role for the construction of minimal codes in later sections.
 \item A general primary construction of minimal linear codes from hyperplanes   was derived in Theorem \ref{thm:projC-D}.
 With these general results, minimal codes with new parameters and explicit weight distributions were obtained.

 \item A general secondary construction of minimal linear codes
    was presented and documented in Theorem \ref{thm:D1-D2}.
    From this construction, many minimal codes with $\frac{w_{\min}}{w_{\max}}\le \frac{q-1}{q}$ were produced.

   \end{itemize}

As observed, the geometric depiction of minimal codes via cutting blocking sets is very effective for analyzing and constructing minimal linear codes.
Other good minimal codes may be obtained from the generic constructions  proposed in this paper, a lot of work can be done in this direction.
It would be nice if the open problems presented in this paper could be settled.
The reader is cordially invited to attack these problems.

\section*{Acknowledgements}
C. Tang was supported by National Natural Science Foundation of China (Grant No.
11871058) and China West Normal University (14E013, CXTD2014-4 and the Meritocracy Research
Funds).

\end{document}